\documentclass[a4paper,english,pagebackref]{elsarticle}
\usepackage[utf8]{inputenc}
\usepackage{amsmath,amssymb,amsthm}
\usepackage[textsize=scriptsize]{todonotes}
\usepackage{paralist}
\usepackage{booktabs}
\usepackage{tabularx}
\usepackage{multirow}

\makeatletter
\def\NAT@spacechar{~}%
\makeatother
\usepackage{tikz}
\usetikzlibrary{decorations.pathmorphing}
\usetikzlibrary{decorations.pathreplacing}
\usetikzlibrary{shapes,automata}
\usetikzlibrary{arrows,calc}

\pgfdeclarelayer{background}
\pgfsetlayers{background,main}

\usepackage{color}
\usepackage{graphicx}
\definecolor{mygrey}{rgb}{0.9,0.9,0.9}
\definecolor{darkgreen}{RGB}{0,100,0}

\usepackage[pagebackref,pdfdisplaydoctitle,menucolor=orange!40!black,filecolor=magenta!40!black,urlcolor=blue!40!black,linkcolor=red!40!black,citecolor=green!40!black,colorlinks]{hyperref}
\hypersetup{pdftitle={The Complexity of Length- and Neighborhood-Constrained Paths}, pdfauthor={Max-Jonathan Luckow and Till Fluschnik}}

\usepackage[sort&compress,nameinlink,noabbrev,capitalize]{cleveref}
\usepackage[T1]{fontenc}

\usepackage{etoolbox}

\theoremstyle{plain}
\newtheorem{theorem}{Theorem}
\newtheorem{corollary}{Corollary}

\newtheorem{lemma}{Lemma}
\newtheorem{proposition}{Proposition}

\theoremstyle{definition}

\crefname{construction}{Construction}{Constructions}
\crefname{step}{Step}{Steps}
\crefname{observation}{Observation}{Observations}
\Crefname{observation}{Obs.}{Obs.}
\crefname{lemma}{Lemma}{Lemmas}
\Crefname{lemma}{Lem.}{Lem.}
\crefname{theorem}{Theorem}{Theorems}
\Crefname{theorem}{Thm.}{Thm.}
\crefname{proposition}{Proposition}{Propositions}
\Crefname{proposition}{Prop.}{Props.}
\crefname{remark}{Remark}{Remarks}
\crefname{prop}{Property}{Properties}
\crefname{corollary}{Corollary}{Corollaries}
\Crefname{corollary}{Cor.}{Cors.}
\creflabelformat{prop}{(#2#1#3)}

\newcommand{\yes}{\textsc{yes}}

\newcommand{\Oh}{\ensuremath{\mathcal{O}}}

\newcommand{\N}{\mathbb{N}}
\newcommand{\I}{\mathcal{I}}

\newcommand{\FPT}{\ensuremath{\operatorname{FPT}}}
\newcommand{\XP}{\ensuremath{\operatorname{XP}}}
\newcommand{\W}[1]{\ensuremath{\operatorname{W}[#1]}}
\newcommand{\poly}{\ensuremath{\operatorname{poly}}}
\newcommand{\NP}{\ensuremath{\operatorname{NP}}}
\newcommand{\coNP}{\ensuremath{\operatorname{coNP}}}
\newcommand{\unlessPK}{\ensuremath{\coNP\subseteq \NP/\poly}}

\newcommand{\decprob}[3]{
    \begin{center}
    \begin{minipage}{0.975\textwidth}
    \medskip\noindent
    \textsc{#1} \\
    \noindent\textbf{Input:} #2 \\
    \noindent\textbf{Question:} #3 \medskip\\
    \end{minipage}
    \end{center}
}

\newcommand{\appsymb}{$\star$}
\newcommand{\appref}[1]{{\hyperref[proof:#1]{\appsymb}}}

\newcommand{\gprob}{$\{S,L\}\times\{S,U\}$~\textsc{Path}}
\newcommand{\gprobTsc}{\textsc{\gprob}}

\newcommand{\sspTsc}{\textsc{Short Secluded Path}}
\newcommand{\sspAcr}{\textsc{SSP}}

\newcommand{\lspTsc}{\textsc{Long Secluded Path}}
\newcommand{\lspAcr}{\textsc{LSP}}
\newcommand{\supTsc}{\textsc{Short Unsecluded Path}}
\newcommand{\supAcr}{\textsc{SUP}}
\newcommand{\lupTsc}{\textsc{Long Unsecluded Path}}
\newcommand{\lupAcr}{\textsc{LUP}}

\newcommand{\calI}{\mathcal{I}}

\newcommand{\LD}{\smallskip\noindent($\Leftarrow$)}
\newcommand{\RD}{\smallskip\noindent($\Rightarrow$)}

\newcommand{\constr}{\smallskip\noindent\emph{Construction}}
\newcommand{\corr}{\smallskip\noindent\emph{Correctness}}

\newcommand{\thetitle}{On the Computational Complexity of Length- and Neighborhood-Constrained Path Problems}

\begin{document}

\begin{frontmatter}
 
\title{\thetitle{}%
}

\author[a]{Max-Jonathan Luckow\fnref{fn1}}
\ead{mj.luckow@campus.tu-berlin.de}
\author[a]{Till~Fluschnik\corref{cor1}\fnref{fn2}}
\ead{till.fluschnik@tu-berlin.de}

\address[a]{Algorithmics and Computational Complexity, Faculty IV, TU~Berlin, Germany}

\cortext[cor1]{Corresponding author}
\fntext[fn1]{This work is based on the first author's bachelor thesis~\cite{Luckow17}.}
\fntext[fn2]{Supported by the DFG, project DAMM (NI~369/13-2) and TORE (NI~369/18).}

\begin{abstract}
Finding paths in graphs is a fundamental graph-theoretic task.
In this work, we we are concerned with finding a path with some constraints on its length and the number of vertices neighboring the path, that is, being outside of and incident with the path.
Herein, we consider short and long paths on the one side, and small and large neighborhoods on the other side---yielding four decision problems.
We show that all four problems are NP-complete, even in planar graphs with small maximum degree.
Moreover, we study all four variants when parameterized by a bound~$k$ on the length of the path, by a bound~$\ell$ on the size of neighborhood, and by~$k+\ell$.
\end{abstract}


\end{frontmatter}
\thispagestyle{empty}

\section{Introduction}

Finding short or long paths (between two designated terminal vertices) in a graph are fundamental algorithmic problems. 
While a short path can be found in polynomial time by folklore results, computing a long path is an NP-complete problem.
In this work, we study short and long paths with small and large \emph{open} neighborhoods.
The open neighborhood of a path consists of all vertices that are not contained in the path but adjacent to at least one vertex in the path.
Formally, we study the following $2\times 2$ problems:

\decprob{\gprob}
	{An undirected graph~$G$ and two integers~$k\geq1$,~$\ell\geq0$.}
	{Is there a simple non-empty path~$P$ in~$G$ with open neighborhood $N:=|N_G(V(P))|$ such that
	
	\smallskip\noindent
	  \begin{tabularx}{0.575\textwidth}{rp{0.3\textwidth}}
	   \sspTsc~(\sspAcr): & $|V(P)|\leq k$ and $N\leq \ell$? \\
	   \lspTsc~(\lspAcr): & $|V(P)|\geq k$ and $N\leq \ell$? \\
	   \supTsc~(\supAcr): & $|V(P)|\leq k$ and $N\geq \ell$? \\
	   \lupTsc~(\lupAcr): & $|V(P)|\geq k$ and $N\geq \ell$? 
	  \end{tabularx}
	}
	
We also consider their so-called $st$-variants: 
Herein, two distinct vertices~$s$ and~$t$ are part of the input, and the question is whether there is an $st$-path fulfilling the respective conditions.
Note that herein~$k\geq 2$, as at least~$s$ and~$t$ must be contained in the path.
We indicate the $st$-variants by using $st$ as prefix.

\paragraph{Our Contributions}
Our results are summarized in~\cref{results-table}.
We prove \gprob{} (and their $st$-variants) to be \NP-complete even on planar graphs with maximum degree five (seven). 
In all but two cases, we settle the parameterized complexity of the four problems regarding their parameters number~$k$ of vertices in the path and size~$\ell$ of the open neighborhood of the path.  
Regarding the parameter~$k$, we have containment in~\XP{} for the ``short'' variants, and para-\NP-hardness (i.e.\ \NP-hardness for some constant parameter value) for the ``long'' variants.
However, the ``short'' variants are W-hard when parameterized by~$k$ and hence presumably fixed-parameter intractable.
The only cases in which we identified fixed-parameter tractability are for the ``short'' variants when parameterized by the combined parameter~$k+\ell$.
Complementing this, we prove that unless $\unlessPK$ no problem kernel of size polynomial in~$k+\ell$ exists, even in planar graphs with small maximum degree.
Regarding the parameter~$\ell$, we found in three of the four cases para-\NP-completeness.

\newcommand{\smtab}[1]{{\scriptsize#1}}
\renewcommand{\arraystretch}{1.25}
\begin{table*}[t]
  \centering
  \caption
  {
    Overview of our results: NP-c., W[1]/W[2]-h., p-NP-h., noPK~abbreviate NP-complete, W[1]/W[2]-hard, para-NP-hard, no polynomial kernel, respectively.
    $^a$\,(even on planar graphs, \Cref{thm:npcomplete})~~$^b$\,(even on planar graphs with maximum degree seven, \Cref{obs:noPK})
  }
  \begin{tabular}{@{}rllll@{}}  \toprule
	Problem		& \multicolumn{1}{l}{Complexity} & \multicolumn{3}{l}{Parameterized Complexity} \\
			& 		& $k$ 		& $\ell$ 	& $k+\ell$ \\\cmidrule{1-5}
	($st$-)\sspAcr 	& NP-c.$^a$ 	& XP, \W{1}-h.~\smtab{(\Cref{thm:w1hard}})	& p-NP-h.$^a$ 	& FPT~\smtab{(\Cref{cor:sspFPTkell})}/noPK$^b$ \\
	($st$-)\lspAcr 	& NP-c.$^a$ 	& p-NP-h.$^a$ 	& p-NP-h.$^a$ 	& p-NP-h.$^a$ \\
	($st$-)\supAcr 	& NP-c.$^a$ 	& XP, \W{2}-h.~\smtab{(\Cref{thm:w2hardness}}) 	& \emph{open}	& {FPT~\smtab{(\Cref{cor:supFPTkell})}/noPK$^b$} \\
	($st$-)\lupAcr 	& NP-c.$^a$ 	& p-NP-h.$^a$ 	& p-NP-h.$^a$ 	& \emph{open}/noPK$^b$  \\
  \bottomrule	
  \end{tabular}
  \label{results-table}
\end{table*}

\paragraph{Related Work}
\citet{ChechikJPP17} introduced the \textsc{Secluded Path} problem that, other than~\sspAcr{}, seeks to minimize the \emph{closed} neighborhood of the path in question, 
where the closed neighborhood are all vertices that are contained in the path or adjacent to a vertex in the path.
They proved~\textsc{Secluded Path} to be~\NP-hard on weighted or directed graphs of maximum degree four, and to be polynomial-time solvable in undirected unweighted graphs (note that we prove~\sspAcr{} to be~\NP-complete in this case).

\citet{FominGKK17}, building upon the work of \citet{ChechikJPP17}, studied the parameterized complexity of \textsc{Secluded Path} (in its weighted version).
They prove \textsc{Secluded Path} to be~\W{1}-hard when parameterized by the length of the path (which refers to the value~$k-1$ in \sspAcr{}). 
Moreover, they prove \textsc{Secluded Path} to be in FPT when parameterized by the size of the closed neighborhood of the path (which refers to the value~$k+\ell$ in \sspAcr{}), but admits no polynomial kernel when parameterized by the combined parameter size of the closed neighborhood of the path, treewidth and maximum degree of the underlying graph, unless~$\unlessPK$.
We remark that in the proofs of our related results~(\cref{thm:w1hard} and \cref{obs:noPK}), we use ideas similar to those of~\citet{FominGKK17}.

Van Bevern et al.~\cite{BevernFMMSS18} studied the problems of finding $st$-separators with small closed neighborhood (``secluded'') and of finding small $st$-separators with small open neighborhood (``small secluded'').
They motivated to distinguish between the size of the subgraph in question and the size of the open neighborhood.
In addition, they studied several other classical optimization problems in their ``secluded'' and ``small secluded'' variant.
Moreover, they also studied the~\textsc{Independent Set} problem, being a maximization problem, in its ``large secluded'' variant.

\citet{GolovachHLM17} studied the problem of finding connected secluded induced subgraphs.
They prove that the problem of finding a connected secluded $\mathcal{F}$-free subgraph is in FPT when parameterized by the size~$\ell$ of the (open) neighborhood.
They also prove that \lspTsc{} parameterized by~$\ell$ (which we prove to be para-NP-hard) is contained in FPT when asking for an induced path.

Van Bevern et al.~\cite{BevernFT18} studied data reduction and preprocessing for $st$-\sspAcr{} regarding several (structural) parameters.
Amongst others they prove that our kernelization lower bound (\cref{obs:noPK}) also holds true when additionally combining with the treewidth of the graph.

\section{Preliminaries}
\label{sec:prelims}
\noindent
We use basic notation from graph theory~\cite{Diestel10} and parameterized complexity theory~\cite{DowneyF13,CyganFKLMPPS15}.

A path~$P$ of length~$\ell-1$ is a graph with vertex set~$\{v_1,\ldots,v_{\ell}\}$ and edge set~$\{\{v_i,v_{i+1}\}\mid 1\leq i< \ell\}$.
We call~$v_1$ and~$v_{\ell}$ the endpoints of~$P$, and hence also refer to~$P$ as a~$v_1$-$v_{\ell}$~path.
For a graph~$G=(V,E)$, we denote by~$N_G(W):=\{v\in V\setminus W\mid \exists w\in W: \{v,w\}\in E\}$ for any~$W\subseteq V$ the~\emph{open} neighborhood of~$W$ in~$G$.
We denote by~$\deg_G(v)=|N_G(v)|$ the degree of vertex~$v$ in~$G$, and by~$\Delta(G)=\max_{v\in V(G)}\deg_G(v)$ the maximum vertex-degree of~$G$.
We say that a path~$P$ is \emph{$k$-short} (\emph{$k$-long}) if~$|V(P)|\leq k$ ($|V(P)|\geq k$).
We say that a path~$P$ is \emph{$\ell$-secluded} (\emph{$\ell$-unsecluded}) if~$|N_G(V(P))|\leq \ell$ ($|N_G(V(P))|\geq \ell$).
A parameterized problem~$L$ is fixed-parameter tractable (in FPT) if for each instance~$(x,p)$ in~$f(p)\cdot |x|^{\Oh(1)}$ time one can decide whether~$(x,p)\in L$, for some computable function~$f$ only depending on~$p$.
If~$L$ is~$\W{1}$- or~$\W{2}$-hard, then it is presumably not fixed-parameter tractable.
A problem kernelization for a parameterized problem~$L$ is a polynomial-time algorithm that maps any instance~$(x,p)$ to an instance~$(x',p')$ (the problem kernel) such that~$(x,p)\in L$ if and only if~$(x',p')\in L$ and~$|x'|+p'\leq f(p)$ for some computable function~$f$ only depending on~$p$.
If~$f$ is a polynomial in~$p$, then~$(x',k')$ is a polynomial problem kernel.

\section{NP-completeness}

We show that even in planar graphs with small maximum degree, each of \gprob{} is \NP-complete, 
in some cases even if the requested size of the path or the size of the open neighborhood is constant.

\begin{theorem}
 \label{thm:npcomplete}
 The following problems are \NP-complete, even on planar graphs:
 \begin{compactenum}[(a)]
  \item \sspAcr{} even if~$\ell=0$ and~$\Delta=3$;
  \item \lspAcr{} even if~$\ell=0$,~$\Delta=3$, and~$k= 1$;
  \item \supAcr{} even if~$\Delta= 5$;
  \item \lupAcr{} even if~$\ell=0$ and~$\Delta= 3$, or~$k= 1$.
 \end{compactenum}
\end{theorem}
In the proof of~\cref{thm:npcomplete}, we give many-one reductions from the \NP-complete~\cite{GareyJT76} 
\textsc{Planar Cubic Hamiltonian Path (PCHP)} problem:
Given an undirected, planar, cubic, connected graph $G=(V,E)$, the question is whether there is a path in $G$ that contains every vertex in~$V$ exactly once.

\begin{proof}%
 The containment in \NP{} is immediate.
 Let $(G)$ be an instance of \textsc{PCHP}.
 Let $G'$ denote a copy of~$G$.
 Denote by~$G''$ the graph obtained from~$G'$ by adding for each vertex~$v\in V$ two vertices to~$G'$ and making them adjacent only with~$v$.
 
 \emph{(a) \& (d)}: Construct the instance $(G',k=n,\ell=0)$.
 On the one hand, note that no neighboring vertices are allowed, hence $G'$ admits a $k$-short $\ell$-secluded path if and only if $G$ admits a Hamiltonian path.
 On the other hand, note that all vertices are required to be contained in the path, hence~$G'$ admits a $k$-long $\ell$-unsecluded path if and only if $G$ admits a Hamiltonian path.
 
 \emph{(b)}: Construct the instance $(G',k=1,\ell=0)$.
 Observe that no neighboring vertices are allowed, hence~$G'$ admits a $k$-long $\ell$-secluded path if and only if $G$ admits a Hamiltonian path.
 
 \emph{(c)}: Construct the instance $(G'',k=n,\ell=2n)$.
 Observe that every path with at least~$2n$ neighbors needs to contain all the vertices in~$G'$, as every path of length~$1\leq r\leq n$ has at most~$2r+(n-r)=n+r\leq 2n$ neighbors.
 Hence, $G''$ admits a $k$-short $\ell$-unsecluded path if and only if $G$ admits a Hamiltonian path.
 
 \emph{(d)}: Construct the instance $(G'',k=1,\ell=2n)$.
 Again, similar to~(c), every path with at least~$2n$ neighbors needs to contain all the vertices in~$G'$.
 Analogously, $G'$ admits a $k$-long $\ell$-unsecluded path if and only if~$G$ admits a Hamiltonian path.
\end{proof}

Next we prove that the $st$-variants are \NP-complete in the same restricted cases, that is, on planar graphs of small maximum degree.

\begin{theorem}
 Even on planar graphs with~$s$ and~$t$ being on the outerface, the following problems are \NP-complete:
 \begin{compactenum}[(a)]
  \item $st$-\sspAcr{} and $st$-\lupAcr{} for any constant~$\ell\geq 0$ and~$\Delta\geq 4$;
  \item $st$-\lspAcr{} for any constant~$\ell\geq 0$,  $k\geq 2$, and~$\Delta\geq 4$.
  \item $st$-\supAcr{} for any constant~$\Delta\geq 6$; 
  \item $st$-\lupAcr{} for any constant~$\ell\geq 0$ and~$\Delta\geq 4$, or~$k\geq 1$ and $\Delta\geq 6$.
 \end{compactenum}
\end{theorem}

\begin{proof}
 We give a many-one reduction from the \NP-complete~\cite{GareyJT76} \textsc{Planar Cubic Hamiltonian Cycle (PCHC)}, which is PCHP where instead of asking for a Hamiltonian path, one asks for a Hamiltonian cycle.

 Let $\calI=(G)$ be an instance for PCHC, and let $c\in\N\cup\{0\}$ be a constant.
 We construct an instance $\calI'=(G',s,t,k,\ell)$ as follows.
 Let $G''$ denote a copy of~$G$, and let~$G'$ initially be~$G''$.
 Consider a plane embedding of~$G$ such that $x$,$y$,$z$ are incident to the outerface and $y,z$ are neighbors of~$x$.
 We add $s$ and $t$ to~$G'$, as well as the edges $\{s,x\}$ and $\{y,t\},\{z,t\}$.
 Next, we add a set~$Z$ of $c$ vertices to~$G'$ and make each vertex in~$Z$ adjacent only with~$s$.
 This finishes the construction of~$G'$.
 Finally, we set $k=n+2$ and $\ell=c$.
 This finishes the construction of~$\calI'$.
 We exemplify the correctness via \emph{(a)} $st$-SSP.
 
 Let~$G$ admit a Hamiltonian cycle~$C$.
 As~$G$ is cubic, vertex~$x$ has three neighbors including~$y$ and~$z$, at least one of them is connected to~$x$ in the cycle~$C$.
 Assume it is~$y$ (for~$z$ the arguments work analogously).
 Then, we construct an $st$-path~$P$ as follows.
 We set~$V(P):=\{s,t,V(C)\}$.
 Next, we set~$E(P):=\{\{s,x\},\{y,t\}\}\cup (E(C)\setminus \{x,y\})$.
 Intuitively,~$P$ is starting at~$s$, going to~$x$ and following cycle~$C$, starting at the neighbor of~$x$ not being~$y$, and ending at~$y$, and finally taking the edge from~$y$ to~$t$.
 Clearly,~$P$ is an~$st$-path and contains~$n+2$ vertices.
 Moreover,~$P$ is also~$\ell$-secluded as~$N_{G'}(V(P))=Z$.
 It follows that~$\I'$ is a \yes-instance.
 
 Conversely, let $G'$ admit a $k$-short $\ell$-secluded $st$-path~$P$.
 Note that~$Z\subseteq N_{G'}(V(P))$ (and hence~$V(P)\cap Z=\emptyset$).
 Moreover, since~$|Z|=c$, $Z= N_{G'}(V(P))$, that is,~$P$ has no neighbors outside of~$Z$.
 Since~$x\in V(P)$, $P$ must contain all vertices in~$V(G'')$.
 Moreover, $P$ needs to contain either edge~$\{y,t\}$ or edge~$\{z,t\}$.
 Assume it is~$\{y,t\}$ (for~$\{z,t\}$ the arguments work analogously).
 Let~$P'\subseteq P$ be the subpath of~$P$ with~$V(P')=V(G'')$.
 Then~$C:=(V(P'),E(P')\cup\{y,x\})$ forms a cycle in~$G''$ containing every vertex exactly once.
 It follows that~$\I$ is a \yes-instance.
 
 \emph{(b)}: Observe that for $st$-\lspAcr{}, $\ell=c$ forces the path to visit all vertices in~$G$, and hence even with~$k=2$ the statement follows.
 
 \emph{(c)}: For~$st$-\supAcr{}, modify~$\I'$ as follows.
 For each vertex in~$v\in V(G'')$, add two vertices and make them adjacent only with~$v$.
 Denote the obtained graph by~$G_+$.
 Set~$\ell'=2n+c$.
 Consider the instance~$\I_+:=(G_+,s,t,k,\ell')$.
 With the same arguments as in the proof of~\cref{thm:npcomplete}, the statement follows.
 Note here that~$\Delta(G_+)=\max\{c,6\}$.
 
 \emph{(d)}: For $st$-\lupAcr{},~$k=n+2$ forces the path to visit all vertices in~$G$.
 Additionally consider the instance~$\I_+'=(G_+,s,t,k',\ell')$ with~$k'\geq 2$, where $\ell'=2n+c$ again forces the path to visit all vertices in~$G$.
\end{proof}

\section{Parameterized Complexity}

In this section, we study the parameterized complexity of ($st$-)\gprobTsc{}.
In~\cref{ssec:upperbounds}, we prove~$st$-\sspAcr{} and~$st$-\supAcr{} to be in \FPT{} when parameterized by~$k+\ell$, and to be in~$\XP$ when parameterized by~$k$.
In~\cref{ssec:lowerbounds}, we prove fixed-parameter intractability regarding~$k$ for \sspAcr{} and~\supAcr{}, 
and prove that~($st$-)\gprobTsc{} admits no problem kernel of size polynomial in~$k+\ell$ unless~\unlessPK{}.

We first relate \gprobTsc{} and their $st$-variants.

\begin{lemma}
  \label{lem:manyoneredtost}
 There is a many-one reduction that maps any instance~$(G,k,\ell)$ of \gprobTsc{} in polynomial time to an instance~$(G',s,t,k',\ell')$ of its $st$-variant such that 
    $k'=k+2$ and~$\ell'=2({|V(G)|\choose2}-1)+\ell$. 
\end{lemma}

\begin{proof}
Given a non-trivial instance~$\I:=(G,k,\ell)$, construct instance~$\I':=(G',s,t,k',\ell')$ as follows.
Let~$G'$ initially only consist of the (isolated) vertices~$s$ and~$t$.
Next, for each pair~$\{v,w\}\in \binom{V(G)}{2}$, add a copy~$G_{\{v,w\}}$ of~$G$ to~$G'$ and make~$s$ adjacent with~$v$ and~$t$ adjacent with~$w$. 
Observe that every~$st$-path in~$G'$ must contain---except for~$s$ and~$t$---only vertices in exactly one copy of~$G$ in~$G'$.
Hence, every~$st$-path has~$2({|V(G)|\choose2}-1)$ unavoidable neighbors.

Let~$\I$ be a \yes-instance and let~$P$ be a path with endpoints~$v$ and~$w$ forming a solution to~$\I$.
Consider the copy~$G_{\{v,w\}}$ in~$G'$, and let~$P^*$ denote the copy of~$P$ in~$G_{\{v,w\}}$.
Then the path~$P'=(V(P^*)\cup\{s,t\},E(P^*)\cup\{\{s,v\},\{w,t\}\})$ in~$G'$ forms a solution to~$\I'$ as $||V(P)|-|V(P')||=|k-k'|=2$ and $||N_G(V(P))|-|N_{G'}(V(P'))||=|\ell-\ell'|=2({|V(G)|\choose2}-1)$.

Let~$\I'$ be a \yes-instance and let~$P'$ be a path forming a solution to~$\I'$.
Let~$N_{P'}(s)=\{v\}$ and~$N_{P'}(t)=\{w\}$.
Let~$G_{\{v,w\}}$ be the copy of~$G$ with~$V(G_{\{v,w\}})\cap V(P')\neq \emptyset$.
Let~$P^*$ denote~$P'$ restricted to~$G_{\{v,w\}}$, that is, the path with vertex set~$V(G_{\{v,w\}})\cap V(P')$.
Let~$P$ be the copy of~$P^*$ in~$G$.
We have $||V(P)|-|V(P')||=|k-k'|=2$ and $||N_G(V(P))|-|N_{G'}(V(P'))||=|\ell-\ell'|=2({|V(G)|\choose2}-1)$, and hence,~$P$ forms a solution to~$\I$.
\end{proof}

As the many-one reduction given in~\cref{lem:manyoneredtost} is also a parameterized reduction regarding the solution size~$k$, we get the following.

\begin{corollary}
  \label{cor:whardnesstransfer}
 If \gprob{} is \W{i}-hard with respect to~$k$, then its $st$-variant is \W{i}-hard with respect to~$k$, for every $i\geq 1$.
\end{corollary}

On the contrary, with a similar idea as in~\cref{lem:manyoneredtost}, one can see that positive results for the~$st$-variants propagate to their counter parts.

\begin{lemma}
  \label{obs:turredtost}
  If any instance $\I=(G,s,t,k,\ell)$ of~$st$-\gprobTsc{} can be decided in~$f(|\I|,k,\ell)$-time,
  then instance~$(G,k,\ell)$ of~\gprobTsc{} can be decided in $\Oh(|V(G)|^2 f(|\I|,k,\ell))$-time.
\end{lemma}
\begin{proof}
  Let~$\I:=(G,k,\ell)$ be a non-trivial instance with~$k>1$.
  We can test for each candidate pair for~$s$ and~$t$, that is, we decide the~$st$-variant on instance~$(G',s,t,k,\ell)$ for every~$\{s,t\}\in\binom{V(G)}{2}$, where~$G'$ denotes a copy of~$G$.
  Observe that~$\I$ is a \yes-instance if and only if there is at least one~$\{s,t\}\in\binom{V(G)}{2}$ such that $(G',s,t,k,\ell)$ is a \yes-instance.
\end{proof}
Due to~\cref{obs:turredtost} we obtain the following.

\begin{corollary}
  \label{cor:FPTtransfer}
 \gprobTsc{} is in FPT when parameterized by~$k$ and/or by~$\ell$ whenever its~$st$-variant is in~FPT when parameterized by~$k$ and/or by~$\ell$, respectively.
\end{corollary}

Due to~\cref{cor:whardnesstransfer}, we prove negative results for the general versions.
Due to~\cref{cor:FPTtransfer}, we prove positive results for the $st$-variants.

\subsection{Upper Bounds}
\label{ssec:upperbounds}

In this section, we prove $st$-\sspAcr{} and~$st$-\supAcr{} to be in \FPT{} when parameterized by~$k+\ell$, and obtain from these results containment in~$\XP$ regarding~$k$.
To this end, we prove the following.

\begin{proposition}
 \label{prop:degreetok}
 Let~$G=(R\uplus B,E)$ be a graph with two distinct vertices~$s,t\in B$,~$k\geq 2,l\geq 0$ be two integers, and~$\Delta_B:=\Delta(G[B])$.
 Then, in time~$\Oh(\Delta_B^{k})\cdot n^{\Oh(1)}$, we can decide whether there is a~$k$-short $\ell$-secluded or $\ell$-unsecluded~$st$-path only containing vertices in~$B$.
\end{proposition}

\begin{proof}
 Construct the following branching tree~$T$ of height at most~$k-1$ with labeling function~$\kappa:V(T)\to B$ and root~$r$ and~$\kappa(r)=s$.
 Each node~$\alpha$ in~$T$ of depth at most~$k-2$ labeled with a vertex~$v\in B\setminus\{t\}$ has~$\deg_{G[B]}(v)$ children such that there is a bijection between the children and~$N_G(v)\cap B$.
 Each node labeled with vertex~$t$ has no children.
 Observe that~$T$ has at most~$\Delta_B^{k}$ vertices. 
 Check in polynomial time for each root-leaf path~$P=(\{\alpha_i\mid 1\leq i\leq q\},\{\{\alpha_i,\alpha_{i+1}\}\mid 1\leq i\leq q-1\})$, where $2\leq q\leq k$ and~$r=\alpha_1$, 
 whether the graph~$P'=(\bigcup_{ 1\leq i\leq q}\{\kappa(\alpha_i)\},\bigcup_{1\leq i\leq q-1}\{\{\kappa(\alpha_i),\kappa(\alpha_{i+1})\}\})$,
 obtained from the labeling of~$P$,
 forms a~$k$-short $\ell$-secluded or $\ell$-unsecluded path in~$G$.
\end{proof}

Intuitively, we prove for $st$-\sspAcr{} and~$st$-\supAcr{} that we can partition the vertex set of the input graph such that the set containing~$s$ and~$t$ has maximum degree upper bounded in~$k+\ell$ and~$\ell+1$, respectively, and apply~\cref{prop:degreetok} consequently.

\begin{theorem}
  \label{cor:sspFPTkell}
 $st$-\sspAcr{} can be solved in~$\Oh((k+\ell)^{k})\cdot n^{\Oh(1)}$~time and hence is in~\FPT{} when parameterized by~$k+\ell$.
\end{theorem}

\begin{proof}
 Let $(G=(V,E),s,t,k,\ell)$ be an instance of $st$-\sspAcr{}.
 Partition~$V=R\uplus B$ such that~$R:=\{v\in V\mid \deg(v)\geq k+\ell+1\}$.
 Clearly, no $k$-short $\ell$-secluded $st$-path can contain any vertex from~$v$.
 Apply~\cref{prop:degreetok} with partition~$R\uplus B$ and~$k,\ell$.
\end{proof}

For~$st$-SUP, tractability also holds true.
\begin{theorem}
  \label{cor:supFPTkell}
 $st$-\supAcr{} can be solved in~$\Oh((\ell+1)^{k})\cdot n^{\Oh(1)}$~time and hence is in~\FPT{} when parameterized by~$k+\ell$.
\end{theorem}

\begin{proof}
 Let~$\I=(G=(V,E),s,t,k,\ell)$ be an arbitrary but fixed input instance to~$st$-\supAcr{}.
 Our FPT-algorithm consists of two phases.
 
 In the first phase, compute the vertex set~$R=\{v\in V\mid \deg(v)\geq \ell+2\}$.
 For each~$v\in R$, do the following.
 Check whether there is a $k$-short $st$-path containing~$v$.
 If yes, then return \yes{} as~$\I$ is a \yes-instance: 
 There is a $k$-short $st$-path (of minimal length) containing~$v$ having at least~$\ell$ neighbors.
 We can check whether there is a $k$-short $st$-path containing~$v$ in polynomial time, by solving the following minimum-cost flow problem.
 
 Construct a directed graph~$D$ as follows.
 Let~$D$ be initially empty.
 First, add a source vertex~$\sigma$ and a sink vertex~$\tau$.
 Next, for each vertex~$w\in V$, add two vertices~$w_{+}$ and~$w_{-}$, as well as the arc~$(w_+, w_-)$ and set the cost and capacity to one.
 For each~$\{u,w\}\in E$, add the two arcs~$(u_-,w_+)$ and~$(w_-,u_+)$, and set for each the cost to zero and the capacity to one.
 Next, add the arcs~$(s_-,\tau)$ and~$(t_-,\tau)$ with cost zero and capacity one.
 Finally, add the arc~$(\sigma,v_-)$ with cost zero and capacity two.
 We denote by~$W$ the set of vertices and by~$A$ the set of arcs of~$D$.
 We claim that~$D$ admits a flow of value two with cost at most~$k-1$ if and only if there is a $k$-secluded $st$-path containing~$v$ in~$G$.
 Note that minimum-cost flow can be solved in polynomial time through e.g.~linear programming.
 
 \RD{}
 Let $D$ admit a flow~$f$ of value two with cost at most~$k-1$.
 As all capacities are integral, we can assume that~$f$ is integral.
 Let~$F=\{(w_+,w_-)\in A\mid f((w_+,w_-))=1\}$.
 Observe that~$|F|\leq k-1$.
 We claim that~$U=\{u\in V\mid (u_+,u_-)\in F\}\cup\{v\}$ forms a $k$-short~$st$-path~$P$ in~$G$ containing~$v$.
 Note that since~$f(a)\in\{0,1\}$ for all~$a\in A\setminus\{(\sigma,v_-)\}$, we can derive from~$f$ a~$v$-$s$ path~$P_s$ on the one hand, and a~$v$-$t$ path~$P_t$ on the other hand.
 Observe that by construction of~$D$,~$P_s$, and~$P_t$ are vertex-disjoint.
 It follows that~$P$ is an~$st$-path~$P$ in~$G$ containing~$v$.
 Finally, as~$|U|=|F|+1\leq k$, we have that~$P$ is also~$k$-short.
 
 \LD{}
 Let~$P$ be a $k$-secluded $st$-path containing~$v$ in~$G$.
 We denote~$V(P)=\{u^1,\ldots,u_{k'}\}$ and~$E(P)=\{\{u^i,u^{i+1}\}\mid 1\leq i<k'\}$, where~$u_1=s$,~$u_{k'}=t$ and~$k'\leq k$.
 Note that there is some index~$x\in[k']$ with~$u^x=v$.
 We construct a function~$f:A\to\{0,1,2\}$ as follows.
 Set~$f((\sigma,v_-)):=2$, $f((s_-,\tau)):=1$, and~$f((t_-,\sigma)):=1$.
 Finally, set
 \[ f((u,u')) :=\begin{cases} 
      1,& \text{if }\exists j\in[k']\setminus\{x\}: (u,u')=(u^j_+,u^j_-) \text{ or}\\
       & \exists j\in\{x,\ldots,k'-1\}: (u,u')=(u^j_-,u^{j+1}_+) \text{ or}\\
       & \exists j\in\{2,\ldots,x\}: (u,u')=(u^j_-,u^{j-1}_+),\\
      0,& \text{otherwise}. 
    \end{cases}
 \]
 Clearly,~$f$ is a~$\sigma$-$\tau$ flow of value two.
 As~$|V(P)|\leq k'\leq k$ and $f$ assigns one to exactly~$k'-1$ arcs of cost one each,~$f$ has cost at most~$k-1$.

 In the second phase, 
 apply~\cref{prop:degreetok} with partition~$R\uplus B$, where~$B:=V\setminus R$ and~$\Delta(G[B])\leq \ell+1$.
\end{proof}

From~\cref{cor:sspFPTkell,cor:supFPTkell}, we immediately obtain the following.

\begin{corollary}
 \label{cor:xpwrtk}
 $st$-\sspAcr{} and $st$-\supAcr{} are contained in~\XP{} when parameterized by~$k$.
\end{corollary}

\subsection{Lower Bounds}
\label{ssec:lowerbounds}
In the previous section, we proved~\sspAcr{} to be solvable in~$\Oh(2^{k\log(k+\ell)})\cdot n^{\Oh(1)}$~time (\cref{cor:sspFPTkell}) and \supAcr{} to be solvable in~$\Oh(2^{k\log(\ell+1)})\cdot n^{\Oh(1)}$~time (\cref{cor:supFPTkell}).
Due to the reductions given in \cref{thm:npcomplete}, assuming that the \emph{Exponential Time Hypothesis (ETH)}~\cite{ImpagliazzoPZ01} holds true, we cannot essentially improve the running times for \sspTsc{} and \supTsc{} regarding the parameter~$k+\ell$.

\begin{corollary}
 Unless the ETH breaks, ($st$-)\sspAcr{} and ($st$-)\supAcr{} are not solvable in~$\Oh(2^{o(k+\ell)})n^{\Oh(1)}$-time.
\end{corollary}

\begin{proof}
 In the many-one reductions given in~\cref{thm:npcomplete}, we have that~$k+\ell\in \Oh(n)$, where~$n$ denotes the number of vertices in the input graph.
 The statement then follows by the fact that~\textsc{Hamiltonian Path} is not solvable in~$\Oh(2^{o(n)})\cdot n^{\Oh(1)}$~time unless the ETH breaks~\cite{CyganFKLMPPS15}.
\end{proof}

Due to~\cref{cor:xpwrtk}, we know that both~\sspAcr{} and \supAcr{} are contained in~\XP{} when parameterized by~$k$.
Our two following results show that containment in~\FPT{} when parameterized by~$k$ only is excluded for~\sspAcr{} (unless~$\FPT=\W{1}$) and for~\supAcr{} (unless~$\FPT=\W{2}$).

\begin{theorem}
 \label{thm:w1hard}
 \sspAcr{} is \W{1}-hard when parameterized by~$k$.
\end{theorem}

In the following proof, we consider the \textsc{Clique} problem: Given an undirected graph~$G$ and an integer~$k\in\N$, decide whether~$G$ contains a~$k$-clique, where a~$k$-clique is a graph on at least~$k$ vertices such that each pair of vertices is adjacent.
\textsc{Clique} parameterized by the solution size~$k$ is a classical \W{1}-complete problem~\cite{DowneyF13,DowneyFVW99}.

\begin{proof}
  Let $(G=(V,E),k)$ be an instance of \textsc{Clique}. 
  In polynomial time, we construct the instance $(G',k',\ell)$ of \sspAcr{} with~$k'=\binom{k}{2}$ as follows (see~\cref{fig:w1hard} for an illustration).
 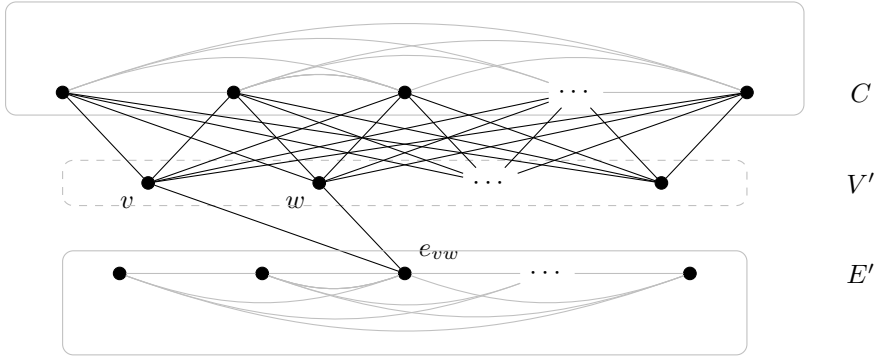
\begin{figure}[t]
  \centering
    \begin{tikzpicture}[scale=1.2]

      \tikzstyle{xn}=[fill, circle, scale=1/2, draw]
      \def\x{1.25}
      \def\y{1}

      \node (a1) at (0*\x,0*\y)[xn]{};
      \node (a2) at (1.5*\x,0*\y)[xn]{};
      \node (a3) at (3*\x,0*\y)[xn]{};
      \node (a4) at (4.5*\x,0*\y)[]{$\cdots$};
      \node (a5) at (6*\x,0*\y)[xn]{};

      \node (lbl) at (7*\x,0*\y)[]{$C$};

      \def\in{160}
      \def\out{20}
      \draw[color=lightgray] (a1) -- (a2) -- (a3) -- (a4) -- (a5);
      \draw[color=lightgray] (a1) to [out=\out,in=\in](a3);
      \draw[color=lightgray] (a1) to [out=\out,in=\in](a4);
      \draw[color=lightgray] (a1) to [out=\out,in=\in](a5);
      \draw[color=lightgray] (a2) to [out=\out,in=\in](a3);
      \draw[color=lightgray] (a2) to [out=\out,in=\in](a3);
      \draw[color=lightgray] (a2) to [out=\out,in=\in](a3);
      \draw[color=lightgray] (a2) to [out=\out,in=\in](a4);
      \draw[color=lightgray] (a2) to [out=\out,in=\in](a5);
      \draw[color=lightgray] (a3) to [out=\out,in=\in](a5);
      \draw[rounded corners, lightgray] (-0.5*\x,-0.25*\y) rectangle (6.5*\x,1*\y);

      \def\ysh{1.0}

      \node (b1) at (0.75*\x,0*\y-\ysh*\y)[xn,label=225:{$v$}]{};
      \node (b2) at (2.25*\x,0*\y-\ysh*\y)[xn,label=225:{$w$}]{};
      \node (b3) at (3.75*\x,0*\y-\ysh*\y)[]{$\cdots$};
      \node (b4) at (5.25*\x,0*\y-\ysh*\y)[xn]{};

      \node (lbl) at (7*\x,0*\y-\y*\ysh)[]{$V'$};

      \draw[dashed, rounded corners, lightgray] (-0.0*\x,-0.25*\y-\ysh*\y) rectangle (6.0*\x,0.25*\y-\ysh*\y);

      \foreach \i in {1,...,5}{
	      \foreach \j in {1,...,4}{
		      \draw (a\i) -- (b\j);
	      }
      }

      \node (a1) at (0.5*\x,0*\y-2*\ysh*\y)[xn]{};
      \node (a2) at (1.75*\x,0*\y-2*\ysh*\y)[xn]{};
      \node (a3) at (3*\x,0*\y-2*\ysh*\y)[xn,label=45:{$e_{vw}$}]{};
      \node (a4) at (4.25*\x,0*\y-2*\ysh*\y)[]{$\cdots$};
      \node (a5) at (5.5*\x,0*\y-2*\ysh*\y)[xn]{};

      \node (lbl) at (7*\x,0*\y-2*\y*\ysh)[]{$E'$};

      \draw (b1) -- (a3) -- (b2);

      \def\in{-160}
      \def\out{-20}
      \draw[color=lightgray] (a1) -- (a2) -- (a3) -- (a4) -- (a5);
      \draw[color=lightgray] (a1) to [out=\out,in=\in](a3);
      \draw[color=lightgray] (a1) to [out=\out,in=\in](a4);
      \draw[color=lightgray] (a1) to [out=\out,in=\in](a5);
      \draw[color=lightgray] (a2) to [out=\out,in=\in](a3);
      \draw[color=lightgray] (a2) to [out=\out,in=\in](a3);
      \draw[color=lightgray] (a2) to [out=\out,in=\in](a3);
      \draw[color=lightgray] (a2) to [out=\out,in=\in](a4);
      \draw[color=lightgray] (a2) to [out=\out,in=\in](a5);
      \draw[color=lightgray] (a3) to [out=\out,in=\in](a5);
      \draw[rounded corners, lightgray] (-0.0*\x,0.25*\y-2*\ysh*\y) rectangle (6*\x,-0.9*\y-2*\ysh*\y);
    \end{tikzpicture}
  \caption{Graph~$G'$ in the construction in the proof of~\cref{thm:w1hard}.
  Each of the sets~$C$ and~$E'$ (enclosed by a solid rectangle) induces a clique.}
  \label{fig:w1hard}
 \end{figure}
 
 \constr{}: Let $G'$ be initially empty.
 We add a copy~$V'$ of $V$ to~$G$ (if $v\in V$, we denote its copy in~$V'$ by $v'$).
 Moreover, for each $e\in E$, we add the vertex $v_e$ to~$G'$ (denote the vertex set by~$E'$).
 If $e=\{v,w\}\in E$, then we add the edges $\{v_e,v'\}$ and $\{v_e,w'\}$ to~$G'$.
 Next, add the vertex set~$C$ consisting of~$|E|+k+1$ vertices to~$G'$, and make~$C$ a clique.
 Finally, make~$E'$ a clique.
 Set $k'=\binom{k}{2}$ and $\ell=|E|-k'+k$.
 This finishes the construction.
 
 \corr{}:
 We prove that $G$ contains a $k$-clique if and only if $G'$ admits a $k'$-short $\ell$-secluded path.

 \RD{}
 Let~$G$ contain a $k$-clique~$G[K]$ with~$|K|=k$ and edge set~$F\subseteq E$.
 Denote by~$K'$ and $F'$ the vertices in~$V'$ and~$E'$ corresponding to~$K$ and $F$, respectively.
 Then construct the $k'$-short $\ell$-secluded path~$P$ as follows.
 Let~$F'=\{e_1,\ldots,e_{|F'|}\}$ be an arbitrary enumeration of the vertices in~$F'$. 
 Construct~$P$ with vertex set~$F'$ and edge set~$\{\{e_i,e_{i+1}\mid 1\leq i\leq |F'|-1\}$.
 Recall that~$E'$ forms a clique, and hence~$P$ can be constructed this way.
 Note that~$P$ contains~$k'=|F'|$ vertices.
 The neighborhood~$N_{G'}(P)$ of~$P$ contains $|E|-k'$ vertices in~$E'$, and $k$ vertices in~$V'$ (recall that $K$ forms a clique in~$G$).
 Hence, $P$ is a $k'$-short $\ell$-secluded path.
 
 \LD{}
 Let $G'$ admit a $k'$-short $\ell$-secluded path~$P$.
 First, observe that~$P$ contains no vertex in~$V'\cup C$, as otherwise $|N_{G'}(P)|\geq |E|+k+1-k'> \ell$, yielding a contradiction.
 Hence, $P$~only contains vertices in~$E'$.
 As $P$ contains~$k'-q$, $0\leq q\leq k$, vertices and $E'$ forms a clique in~$G'$, $|N_{G'}(P)\cap E'|=|E|-k'+q$.
 Since~$0\leq \ell-|N_{G'}(P)\cap E'|\leq k-q$, it follows that for the set~$K'=N_{G'}(P)\cap V'$ it holds true that~$|K'|\leq k-q$.
 Since any graph with~$k-q$ vertices has at most~$\binom{k-q}{2}$ edges, we have that~$\binom{k}{2}-q=|E'|\leq \binom{k-q}{2}$, which only holds true for~$q=0$.
 It follows that the vertex set~$K$ corresponding to~$K'$ forms a $k$-clique in~$G$.
\end{proof}

\begin{theorem}
  \label{thm:w2hardness}
 \supAcr{} is  \W{2}-hard  when parameterized by~$k$.
\end{theorem}

In the following proof, we consider the \textsc{Red-Blue Dominating Set (RBDS)} problem: 
Given an undirected graph~$G=(V=R\uplus B,E)$ and an integer~$k\in\N$, 
decide whether~$G$ contains a \emph{red-blue $k$-dominating set}, where a red-blue~$k$-dominating set is a subset~$V'\subseteq R$ with~$|V'|\leq k$ such that each vertex in~$B$ is adjacent to at least one vertex in~$V'$.
RBDS parameterized by the solution size~$k$ is a \W{2}-complete problem~\cite{DowneyF13,DowneyFVW99}.

\begin{proof}
 Let $(G=(V=R\uplus B,E),k)$ be an instance of RBDS($k$).
 In polynomial time, we construct the instance $(G',k',\ell)$ of \supAcr{} with~$k'=3k+1$ as follows (see~\cref{fig:w2hard} for an illustration).
 \begin{figure}[t]\centering
  \begin{tikzpicture}[scale=1.2]

      \tikzstyle{xn}=[fill, circle, scale=1/2, draw]
      \def\x{1.25}
      \def\y{1}
      \def\ysh{1.25}

      \node (a1) at (0*\x,0*\y-0.25*\ysh*\y)[xn]{};
      \node (a2) at (1.5*\x,0*\y-0.25*\ysh*\y)[xn,label=135:{$b$}]{};
      \node (a3) at (3*\x,0*\y-0.25*\ysh*\y)[xn,label=45:{$b'$}]{};
      \node (a4) at (4.5*\x,0*\y-0.25*\ysh*\y)[]{$\cdots$};
      \node (a5) at (6*\x,0*\y-0.25*\ysh*\y)[xn]{};

      \node (lbl) at (7*\x,0*\y-0.25*\ysh*\y)[]{$B$};
      \draw[rounded corners, dashed, lightgray] (-0.5*\x,-0.25*\y-0.25*\ysh*\y) rectangle (6.5*\x,0.25*\y-0.25*\ysh*\y);
      
      \def\in{150}
      \def\out{30}

      \def\ysh{1.25}

      \node (b1) at (0.0*\x,0*\y-1*\ysh*\y)[xn]{};
      \node (b2) at (1.5*\x,0*\y-1*\ysh*\y)[xn,label=135:{$r$}]{};
      \node (b3) at (3*\x,0*\y-1*\ysh*\y)[xn,label=45:{$r'$}]{};
      \node (b4) at (4.5*\x,0*\y-1*\ysh*\y)[]{$\cdots$};
      \node (b5) at (6*\x,0*\y-1*\ysh*\y)[xn]{};

      \node (lbl) at (7*\x,0*\y-1*\y*\ysh)[]{$R$};

      \draw (b2) -- (a2) -- (b3) -- (a3);

      \draw[rounded corners, dashed, lightgray] (-0.5*\x,0.25*\y-1*\ysh*\y) rectangle (6.5*\x,-0.25*\y-1*\ysh*\y);

      \tikzstyle{lnode}=[fill,circle,scale=1/5]
      \newcommand{\bstar}[3]{
	
		\def\noleaves{8};
		\def\distleaves{0.6};
	
		\node (#1) at (#2,#3)[xn,draw]{};
 		 \foreach \j in {1,...,\noleaves}{
			   \node (a\j) at ($(#1)+ (202.5+\j*120/\noleaves:\distleaves cm)$)[lnode]{};
				\draw (#1) -- (a\j);
 		 }
	}
      \bstar{c1}{0.75*\x}{0*\y-2*\ysh*\y};
      \bstar{c2}{2.25*\x}{0*\y-2*\ysh*\y};
      \node (c3) at (3.75*\x,0*\y-2*\ysh*\y)[]{$\cdots$};
      \bstar{c4}{5.25*\x}{0*\y-2*\ysh*\y};

      \node (lbl) at (7*\x,0*\y-2*\y*\ysh)[]{$U$};

      \draw[dashed, rounded corners, lightgray] (-0.25*\x,-0.25*\y-2*\ysh*\y) rectangle (6.25*\x,0.25*\y-2*\ysh*\y);

      \foreach \i in {1,...,5}{
	      \foreach \j in {1,...,4}{
		      \draw (b\i) -- (c\j);
	      }
      }

  \end{tikzpicture}
  \caption{Graph~$G'$ in the construction in the proof of~\cref{thm:w2hardness}.}
  \label{fig:w2hard}
 \end{figure}
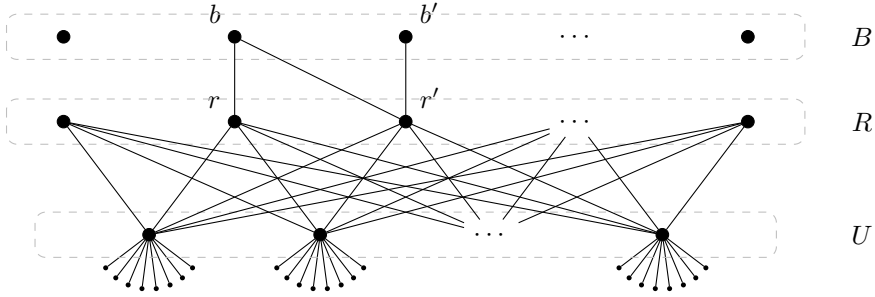
 
 \constr{}: Let $G'$ be initially a copy of~$G$ (denote by~$R'$ and~$B'$ the copies of~$R$ and~$B$ respectively).
 Next add the vertex set~$U=\{u_1,\ldots,u_{k+1}\}$ to~$G'$.
 Connect every vertex in~$R'$ with every vertex in~$U$ via an edge (i.e.~$R'\cup U$ forms a biclique).
 Finally, for each vertex~$u\in U$, add $n^2$~vertices making each adjacent only to~$u$.
 Denote by~$H$ all the vertices introduced in the previous step.
 Set $k'=2k+1$ and $\ell=k\cdot n^2+2n-k$.
 This finishes the construction.
 
 \corr{}: We prove that $G$ admits a red-blue $k$-dominating set if and only if $G'$ admits a $k'$-short $\ell$-unsecluded path.
 
 \RD{} Let $W\subseteq R$ be a red-blue $k$-dominating set in~$G$ with~$|W|=k$.
 Let $W'=\{w_1',w_2',\ldots,w_k'\}\subseteq R'$ denote the vertices in~$R'$ corresponding to the vertices in $W$.
 We claim that the path $P$ with vertex set~$V(P)=\{u_i\mid 1\leq i\leq k+1\}\cup W'$ and edge set~$E(P)=\{\{u_i,w_i'\},\{w_i',u_{i+1}\}\mid 1\leq i\leq k\}$ is a $k'$-short $\ell$-unsecluded path in~$G'$.
 First observe that the number of vertices in~$P$ is $k'=2k+1$.
 As $W$ is a dominating set in~$G$, $N_{G'}(W')=V''\cup U$.
 Moreover, $N_{G'}(U)=R'\cup H$.
 As $P$ consists exactly of the vertices in~$W'\cup U$, we have $|N_{G'}(W'\cup U)|=|N_{G'}(W')|-|U|+|N_{G'}(U)|-|W'|=n+k\cdot n^2 + (n-k)=\ell$.
 
 \LD{} Let $P$ be a $k'$-short $\ell$-unsecluded path in $G'$.
 The first observation is that~$P$ contains all vertices from~$U$ as $P$ has more than~$k\cdot n^2$ neighbors.
 The second observation is that~$P$ needs to alternate between the vertices in~$V'$ and~$U$ as $P$ only contains $k'=2k+1$ vertices and all vertices of~$U$.
 It follows that $P$ contains exactly $k$ vertices~$W'\subseteq R'$.
 As $V(P)\cup N_G'(V(P))=V(G')$, the vertex set~$W'$ dominates all the vertices in~$B'$.
 It follows that the set~$W\subseteq V$ corresponding to~$W'$ forms a red-blue $k$-dominating set in~$G$.
\end{proof}

We proved that when parameterized by~$k+\ell$, \sspAcr{} (\cref{cor:sspFPTkell}) and~\supAcr{} (\cref{cor:supFPTkell}) are contained in~\FPT{}.
We next prove that unless \unlessPK, none of the two problems admits a problem kernel of polynomial size, even on planar graphs with small maximum degree.

\begin{theorem}
  \label{obs:noPK}
 Unless \unlessPK, ($st$-)\gprobTsc{} parameterized by~$k+\ell$ admits no polynomial problem kernel even on planar graphs with maximum degree seven.
\end{theorem}
\noindent

\begin{proof}
 We employ the OR-composition framework~\cite{BodlaenderDFH09}.
 An easy application (taking the disjoint union of the graphs) proves the statement for~\gprobTsc{}.
 Hence, we next consider the $st$-variants.
 Let~$\{\I_i=(G_i,s_i,t_i,k,\ell)\mid 1\leq i\leq p\}$ be a set of~$p$ input instances, where~$p$ is a power of two, and~$G_i$ is planar, is of maximum degree five, and allows for an embedding with~$s,t$ being on the outer face.
 
 ($st$-\sspAcr) We construct the instance~$\I'=(G',s,t,k',\ell')$ as follows.
 Let~$G'$ be initially empty.
 We add two binary trees~$T_s$ and~$T_t$ with root~$s$ and~$t$, respectively, where each tree has~$p$ leaves all being at the same depth.
 Let~$\sigma_1,\ldots,\sigma_p$ denote the leaves of~$T_s$ enumerated through a post-order depth-first search.
 Similarly, let $\tau_1,\ldots,\tau_p$ denote the leaves of~$T_t$ enumerated through a post-order depth-first search.
 Next, for each~$i\in\{1,\ldots,p\}$, add the copy~$G_i'$ of~$G_i$ to~$G'$, and add the edges~$\{\sigma_i,s_i\}$ and~$\{t_i,\tau_i\}$.
 Finally, for each~$i\in\{1,\ldots,p\}$, subdivide the edges~$\{\sigma_i,s_i\}$ and~$\{t_i,\tau_i\}$ each~$k$ times.
 Denote the vertices by~$\sigma_i^1,\ldots,\sigma_i^k$ resulting from the subdivision of~$\{\sigma_i,s_i\}$, enumerated by the distance from~$\sigma_i$, and by~$\tau_i^1,\ldots,\tau_i^k$ resulting from the subdivision of~$\{\tau_i,t_i\}$, enumerated by the distance from~$t_i$.
 For simplicity, we also denote~$\sigma_i$ and~$s_i$ by~$\sigma_i^0$ and~$\sigma_i^{k+1}$, respectively, and $t_i$ and~$\tau_i$ by~$\tau_i^0$ and~$\tau_i^{k+1}$, respectively. 
 This finishes the construction of~$G'$.
 Observe that one can embed both~$T_s$ and~$T_t$ such that when adding the edge set~$\{\{\sigma_i,\tau_i\}\mid 1\leq i\leq p\}$, the resulting graph is crossing-free and~$s$ and~$t$ are on the outer face.
 As each~$G_i$ is planar and allows for an embedding with~$s,t$ being on the outer face, it follows that~$G'$ is planar with~$s$ and~$t$ being on the outer face.
 Moreover, note that~$\Delta(G')\leq 1+\max_{1\leq i\leq p}\Delta(G_i)$.
 Finally, set $k':=3k+2(\log(p)+1)$ and $\ell':=\ell+2\log(p)$.
 We next prove that~$\I'$ is a \yes-instance if and only if there is at least one~$i\in\{1,\ldots,p\}$ such that~$\I_i$ is a \yes-instance.
 
 \LD{}
 Let~$i\in\{1,\ldots,p\}$ such that~$\I_i$ is a \yes-instance, and let~$P$ be a~$k$-short $\ell$-secluded $s_i t_i$-path in~$G$.
 Let~$P'$ denote its copy in~$G_i'$.
 Let~$P_{s,i}$ denote the unique path with endpoints~$s$ and~$\sigma_i$ in~$T_s$.
 Note that~$|V(P_{s,i})|=\log(p)+1$.
 Similarly, let~$P_{t,i}$ denote the unique path with endpoints~$t$ and~$\tau_i$ in~$T_t$.
 Note that~$|N_{T_s}(P_{s,i})|=\log(p)$, as each vertex in~$P_{s,i}$ except~$s$ and~$\sigma_i$ is of degree three in~$T_s$, and~$s$ has one unique neighbor not in~$P_{s,i}$.
 With the same argument, we have $|N_{T_t}(P_{i,t})|=\log(p)$.
 Let~$V_P:=V(P)\cup V(P_{s,i})\cup V(P_{i,t})\bigcup_{j=1}^{k} \{\sigma_i^j,\tau_i^j\}$ and~$E_P:=E(P)\cup E(P_{s,i})\cup E(P_{i,t})\cup \bigcup_{j=0}^{k} \{\{\sigma_i^j,\sigma_i^{j+1}\},\{\tau_i^j,\tau_i^{j+1}\}\}$.
 We claim that the path~$Q=(V_P,E_P)$ is a $k'$-short $\ell'$-secluded $st$-path in~$G'$.
 By construction, ~$Q$ is a~$k'$-short $st$-path in~$G'$.
 Moreover, we have~$|N_{G'}(Q)|=|N_{T_s}(P_{s,i})| + |N_{T_t}(P_{i,t})| + |N_{G_i'}(P')|\leq 2\log(p)+\ell=\ell'$.
 
 \RD{}
 Let~$\I'$ be a \yes-instance, and let~$P$ be a~$k'$-short $\ell'$-secluded $st$-path in~$G'$.
 We claim that there is a subpath~$P'\subseteq P$ such that~$P'$ is a~$k$-short~$\ell$-secluded $s_i t_i$-path in~$G_i$, for some~$i\in\{1,\ldots,p\}$.
 Observe that~$P$ must contain at least one leaf in~$T_s$ and one leaf in~$T_t$.
 Hence,~$|V(P)\cap V(T_s)|\geq \log(p)+1$ and~$|V(P)\cap V(T_s)|\geq \log(p)+1$.
 Moreover, $s_i\in V(P)$ if and only if~$t_i\in V(P)$, as~$P$ has only endpoints~$s$ and~$t$, and~$\{s_i,t_i\}$ separates~$V(G_i)\setminus \{s_i,t_i\}$ from~$V(G')\setminus V(G_i)$.
 Hence, let~$i\in\{1,\ldots,p\}$ such that~$\sigma_i\in V(P)$ (and hence~$\tau_i\in V(P)$).
 Let~$P'$ be the subpath of~$P$ with endpoints~$s_i$ and~$t_i$.
 Clearly,~$V(P')\subseteq V(G_i')$.
 We claim that~$P'$ is a~$k$-short $\ell$-secluded $s_i t_i$-path in~$G_i'$ (and hence, also in~$G_i$).
 First, suppose that~$|V(P')|>k$. 
 Then we have~$|V(P)|\geq |V(P)\cap V(T_s)|+|V(P)\cap V(T_t)|+|V(P')|+2k> 3k+2(\log(p)+1)=k'$, contradicting the fact that $P$ is a~$k'$-short $st$-path in~$G'$.
 Next, we claim that there is no~$j\in\{1,\ldots,p\}\setminus\{i\}$ such that~$s_j\in V(P)$ (and hence,~$t_j\in V(P)$).
 Suppose not. 
 Then~$|V(P)|\geq |V(P)\cap V(T_s)|+|V(P)\cap V(T_t)|+4k> 3k+2(\log(p)+1)=k'$, again contradicting the fact that $P$ is a~$k'$-short $st$-path in~$G'$.
 It follows that~$T_s[V(P)\cap V(T_s)]$ is the unique path in~$T_s$ with endpoints~$s$ and~$\sigma_i$, and~$T_t[V(P)\cap V(T_t)]$ is the unique path in~$T_t$ with endpoints~$t$ and~$\tau_i$.
 Moreover, $|N_{T_s}(V(P))|=|N_{T_t}(V(P))|=\log(p)$.
 Finally, suppose that~$|N_{G_i'}(V(P'))|>\ell$.
 Then we have~$|N_{G'}(V(P))| = |N_{T_s}(V(P))|+|N_{T_t}(V(P))|+|N_{G_i'}(V(P'))|> \ell+2\log(p)=\ell'$, contradicting the fact that $P$ is an~$\ell'$-secluded $st$-path in~$G'$.
 We conclude that $P'$ is a~$k$-short $\ell$-secluded $s_i t_i$-path in~$G_i$, and hence,~$\I_i$ is a \yes-instance.
 
 ($st$-\supAcr)
 The construction is exactly the same as for~$st$-\sspAcr.
 The crucial observation is, again, that every~$k'$-short $\ell'$-unsecluded $st$-path~$P$ in~$G'$ only contains~$s_i$ (and~$t_i$) for exactly one~$i\in\{1,\ldots,p\}$.
 
 ($st$-\lspAcr{})
 Let~$\I'$ as in the construction for~$st$-\sspAcr.
 Make each vertex of the binary trees a star with~$2\log(p)+\ell+1$ leaves, and denote by~$G''$ the graph obtained from~$G'$ in this step.
 Set~$\ell'':=2(\log(p)+1)\cdot(2\log(p)+\ell+1)+\ell+2\log(p)$.
 This forces every~$k'$-long $\ell''$-secluded~$st$-path~$P$ in~$G'$ to only contain~$\log(p)$ vertices in each of the binary trees, as otherwise such a path~$P$ would contain at least~$2(\log(p)+1)\cdot(2\log(p)+\ell+1)+(2\log(p)+\ell+1)>\ell'$ neighbors.
 
 ($st$-\lupAcr{})
 There is an straight-forward polynomial parameter transformation from \textsc{Longest $st$-Path} on planar graphs with maximum degree three~\cite{BodlaenderDFH09}.
 Note that herein, we set~$\ell=0$.
\end{proof}

\section{Conclusion and Outlook}
\label{sec:concl}

All four variants remain \NP-complete in planar graphs with small vertex degree.
However, the ``short'' and ``long'' variants are distinguishable through their parameterized complexity regarding~$k$.
We conjecture that all four variants are pairwise distinguishable through the parameterized complexity regarding the parameters~$k$,~$\ell$, and~$k+\ell$.
To resolve this conjecture, the parameterized complexity of~\supAcr{} parameterized by~$\ell$ and~\lupAcr{} parameterized by~$k+\ell$, that we left open, has to be settled.

As a further research direction, it is interesting to investigate the problem of finding small/large secluded/unsecluded (sub-)graphs different to paths.
For instance, the class of trees could be an interesting next candidate in this context.
Notably, herein the large secluded variant is polynomial-time solvable.

\bibliographystyle{abbrvnat}
\bibliography{cnp-journal}

\end{document}